\documentclass[11pt]{amsart}
\usepackage{bbm,graphicx,enumitem,color,verbatim}

\usepackage[margin=1.25in]{geometry}

\newtheorem{thm}{Theorem}

\newtheorem{lemma}[thm]{Lemma}
\newtheorem{prop}[thm]{Proposition}

\theoremstyle{definition}

\newcommand{\R}{\mathbb{R}}
\renewcommand{\P}{\mathbb{P}}
\newcommand{\E}{\mathbb{E}}

\newcommand{\N}{\mathbb{N}}

\newcommand{\ind}[1]{\mathbbm{1}_{#1}}

\DeclareMathOperator{\tr}{Tr}

\allowdisplaybreaks[3]
\author{David Buzinski}

\author{Elizabeth Meckes}

\address{Department of Mathematics, Case Western Reserve University,
10900 Euclid Ave., Cleveland, Ohio 44106, U.S.A.}

\email{dab197@case.edu}

\address{Department of Mathematics, Case Western Reserve University,
10900 Euclid Ave., Cleveland, Ohio 44106, U.S.A.}

\email{elizabeth.meckes@case.edu}

\title{Almost sure convergence in quantum spin glasses}

\begin{document}

\begin{abstract}
 Recently, Keating, Linden, and Wells \cite{KLW} showed that the density of states
 measure of a nearest-neighbor quantum spin glass model is
 approximately Gaussian when the number of particles is large.  The
 density of states measure is the ensemble average of the empirical
 spectral measure of a random matrix; in this paper, we use
 concentration of measure and entropy techniques together with the
 result of \cite{KLW} to show that in fact,
 the empirical spectral measure of such a random matrix is almost
 surely approximately Gaussian itself, with no ensemble averaging.  We
 also extend this result to a spherical quantum spin glass model and
to the more general coupling geometries investigated by Erd\H{o}s and Schr\"oder.
\end{abstract}

\maketitle


\section{Introduction and statements of results}

In the recent paper \cite{KLW}, Keating, Linden and Wells show that the density of states measure of a quantum spin glass with nearest neighbor interactions and Gaussian coupling coefficients is approximately Gaussian, as the number of particles tends to infinity.  More specifically, they considered the following random matrix model for the Hamiltonian of a quantum spin glass: let $\{Z_{a,b,j}\}_{\substack{1\le a,b\le 3\\1\le j\le n}}$ be independent standard Gaussian random variables, and define the $2^n\times 2^n$ random matrix $H_n$ by
\begin{equation}\label{D:H_n}H_n:=\frac{1}{\sqrt{9n}}\sum_{j=1}^n\sum_{a,b=1}^3Z_{a,b,j}\sigma^{(a)}_j\sigma^{(b)}_{j+1},\end{equation}
where for $1\le a\le 3$,
\[\sigma^{(a)}_j:=I_n^{\otimes(j-1)}\otimes\sigma^{(a)}\otimes I_2^{\otimes(n-j)},\]
with $I_2$ denoting the $2\times 2$ identity matrix, $\sigma^{(a)}$ denoting the $2\times 2$ non-trivial Pauli matrices
\[\sigma^{(1)}:=\begin{bmatrix}0&1\\1&0\end{bmatrix}\qquad \sigma^{(2)}:=\begin{bmatrix}0&-i\\i&0\end{bmatrix}\qquad \sigma^{(3)}:=\begin{bmatrix}1&0\\0&-1\end{bmatrix},\]
and the labeling cyclic so that $\sigma^{(b)}_{n+1}:=\sigma^{(b)}_{1}.$

The density of states measure $\mu_n^{DOS}$ for the system is the ensemble average of the spectral measure of $H_n$; that is, if $\{\lambda_j\}_{1\le j\le 2^n}$ are the (necessarily real) eigenvalues of $H_n$, then for $A\subseteq \R$,
\[\mu_n^{DOS}(A)=\frac{1}{2^n}\E\big|\big\{j:\lambda_j\in A\big\}\big|.\]
In other words, $\mu_n^{DOS}(A)$ is the expected proportion of the eigenvalues of $H_n$ lying in the set $A$.  The main result of \cite{KLW} is that $\mu_n^{DOS}$ converges weakly to Gaussian, as $n\to\infty$.  The authors go on to consider more general collections of (still independent) coupling coefficients, and more general coupling geometries than that of nearest-neighbor interactions.  In more recent work, Erd\H{os} and Schr\"oder \cite{ES} have considered still more general coupling geometries, and found a sharp transition in the limiting behavior of the density of states measure depending on the size of the maximum degree of the underlying graph, relative to its number of edges.

The purpose of this paper is to move from convergence in expectation of the spectral measure of $H_n$ to the considerably stronger notion of almost sure convergece.  As observed in \cite{KLW}, the extent to which ensemble averages actually manage to describe features of individual systems is not always clear, but is a crucial issue if one is to make meaningful use of random matrix models.  
The following result shows that the Gaussian behavior exhibited by the average spectral measure is indeed the typical behavior of the empirical spectral measures.

\begin{thm}\label{T:random-dist-to-gauss}
Let $\mu_n$ be the spectral measure of $H_n$ and let $\gamma$ denote the standard Gaussian distribution.  There are universal constants $C$, $C'$ and $c$ such that
\begin{enumerate}[label=(\alph*)]
\item \label{P:exp-dist-to-gauss}\(\displaystyle\E d_{BL}(\mu_n,\gamma)\le \frac{C}{n^{1/6}};\)
\item \label{P:dist-conc-gauss}\(\displaystyle\P\left[d_{BL}(\mu_n,\gamma)\ge\frac{C}{n^{1/6}}+t\right]\le Ce^{-cnt^2};\)
\end{enumerate}
and
\begin{enumerate}[resume,label=(\alph*)]
\item \label{P:as-dist-to-gauss}with probability 1, for all sufficiently large $n$,
\[d_{BL}(\mu_n,\gamma)\le \frac{C'}{n^{1/6}}.\]
\end{enumerate}
Here $d_{BL}(\mu,\nu)$ denotes the bounded-Lipschitz distance between probability measures, which metrizes the topology of weak convergence.
\end{thm}

In the paper \cite{KLW2}, Keating, Linden and Wells took a different
approach to understanding the behavior of the spectral measures of
individual Hamiltonians; rather than consider a random matrix model,
they took the coefficients in \eqref{D:H_n} to be deterministic,
subject to a normalization and boundedness condition, and showed that
in that case, the non-random spectral measures converged weakly to
Gaussian.  It should be possible to take that result as a starting
point in order to obtain almost sure convergence in the Gaussian
model, although there are various technical challenges.  We instead
take a rather different approach, combining convergence in expectation
with various probabilistic techniques.

The organization of this paper is as follows.   In Section
\ref{S:DOS-to-gamma}, the  pointwise estimate of \cite{KLW} on the difference between characteristic functions of
$\mu_n^{DOS}$ and $\gamma$ is parlayed into an estimate on
$d_{BL}(\mu_n^{DOS},\gamma)$ using Fourier analysis.  In Section
\ref{S:avg-to-avg} the Gaussian concentration of measure phenomenon is
used to show that the random variable $d_{BL}(\mu_n,\gamma)$ is strongly
concentrated at its mean.  Then, the expected distance between $\mu_n$
and its average $\mu_n^{DOS}$ is estimated; this is done using a
combination of applications of Gaussian concentration of measure,
entropy methods, and approximation theory, via a similar approach to
the one taken by the second author and M.\ Meckes in \cite{MM}.  The
almost-sure convergence rate given in part \ref{P:as-dist-to-gauss} is
an immediate consequence of part \ref{P:dist-conc-gauss} and the first
Borel-Cantelli lemma, and is therefore not discussed further.  In
Section \ref{S:sphere}, we consider a modification of the random matrix
model above, in which the coefficients in $H_n$ are not independent
but are drawn uniformly
from the $9n$-dimensional sphere, and show that the empirical spectral
measure is almost surely approximately Gaussian in that setting as
well.  Finally, Section \ref{S:remarks} offers some remarks on extensions of Theorem \ref{T:random-dist-to-gauss} to further related ensembles.

\subsection*{Notation and Conventions}
Let the random matrix $H_n$ be defined as above, with eigenvalues $\lambda_1<\cdots<\lambda_{2^n}$.  The empirical spectral measure $\mu_n$ of $H_n$ is defined by
\[\mu_n:=2^{-n}\sum_{j=1}^{2^n}\delta_{\lambda_j};\]
that is, $\mu_n$ is the random probability measure putting equal mass at each eigenvalue of $H_n$.  Its ensemble average $\E\mu_n$ is denoted $\mu_n^{DOS}$ and is called the density of states measure.

For probability measures $\mu$ and $\nu$ on $\R$, the bounded-Lipschitz distance $d_{BL}(\mu,\nu)$ from $\mu$ to $\nu$ is defined by
\[d_{BL}(\mu,\nu):=\sup\left\{\left|\int fd\mu-\int fd\nu\right|:\|f\|_{BL}\le 1,\right\}\]
where $\|f\|_{BL}$ denotes the bounded-Lipschitz norm of $f$, defined by
\[\|f\|_{BL}:=\|f\|_\infty+|f|_L, \]
with $|f|_L$ denoting the Lipschitz constant of $f$.
The bounded-Lipschitz distance metrizes weak convergence of probability measures.

If the test functions are required only to be Lipschitz and not necessarily bounded, one gets instead the $L_1$-Kantorovich distance $W_1(\mu,\nu)$:
\[W_1 (\mu,\nu):=\sup\left\{\left|\int fd\mu-\int fd\nu\right|:|f|_{L}\le 1\right\}.\]
Clearly, 
\(d_{BL}(\mu,\nu)\le W_1(\mu,\nu).\)

It is the Kantorovich--Rubenstein theorem that $W_1$ is also given by
\[W_1 (\mu,\nu)=\inf_\pi\int|x-y|d\pi(x,y),\]
where the infimum is taken over all couplings $\pi$ of the meausures $\mu$ and $\nu$.  It is for this reason that $W_1$ is also called the $L_1$-coupling distance.

Finally, symbols such as $C,c$ will denote universal constants independent of all parameters, which may change in value from one appearance to the next.
\section{Gaussian density of states in the Kantorovich distance}\label{S:DOS-to-gamma}
The crucial ingredient in the estimation of $d_{BL}(\mu_n^{DOS},\gamma)$ is the following pointwise bound from \cite{KLW} on the difference between the corresponding characteristic functions.
\begin{thm}[Keating, Linden, and Wells]
Let $\mu_n^{DOS}$ and $\gamma$ be as above; denote the characteristic function of $\mu_n^{DOS}$ by $\psi_n$ and the characteristic function of $\gamma$ by $\varphi$.  There is a constant $C$ independent of $n$ such that for all $\xi$,
\begin{equation}\label{E:ft-diff-bound}
\left|\psi_n(\xi)-\varphi(\xi)\right|\le\frac{C\xi^2}{\sqrt{n}}.\end{equation}
\end{thm}

The main result of this section is the following.
\begin{thm}\label{T:W_1_to_gamma}
Let $\mu_n^{DOS}$ and $\gamma$ be as above.  There is a constant $c$ such that  
\[d_{BL}(\mu_n^{DOS},\gamma)\le\frac{c}{n^{1/4}}.\]
\end{thm}

Note that by defintion of the bounded-Lipschitz distance there is no loss in restricting to the case $f(0)=0$, which we do for the remainder of the proof.

The first step is to make a truncation argument to further restrict the class of test functions considered.  Given a function $f:\R\to\R$ with $\|f\|_{BL}\le 1$ and $f(0)=0$, and given $R>0$, define the trunction $f_R$ by 
\begin{equation}\label{E:truncation}f_R(x) = 
	\begin{cases}
       		f(x), &  |x| \le R; \\
       		f(R)+ [sgn(f(R))](R-x), &  R<x<R+|f(R)|; \\
       		f(-R)+ [sgn(f(-R))](x+R), &  -|f(-R)|-R<x<-R; \\
       		0, & x \le -|f(-R)|-R$ or $x \ge R+|f(R)|. \\
     	\end{cases}
	\end{equation}
Then $\|f_R\|_{BL}\le 1$ and $f_R$ is supported on $[-2R,2R]$.

\begin{lemma}\label{T:truncation}~

\begin{enumerate}[label=(\alph*)]
\item For any $t>0$,  \label{P:dos-tails}
\[\mu_n^{DOS}\left(\left\{x:|x|>t\right\}\right)\le \frac{c}{t^2}.\]
\item For 
$f:\R\to\R$ with $\|f\|_{BL}\le 1$,
\[\left|\int (f-f_R)d\mu_n^{DOS}\right|\le \frac{c}{R^2}.\]\label{P:dos-trunc-error}
\end{enumerate}
\end{lemma}

\begin{proof}
A straightforward Fubini's theorem argument (see, e.g., Section 26 of \cite{Bil}) gives that
\[\mu_n^{DOS}\left(\left\{x:|x|>t\right\}\right)\le \frac{t}{2}\int_{-\frac{2}{t}}^{\frac{2}{t}}(1-\psi_n(\xi))d\xi,\]
where as before $\psi_n(\xi)$ denotes the characteristic function of $\mu_n^{DOS}$.  Adding and subtracting the characteristic function $\varphi(\xi)$ of the standard Gaussian distribution and using \eqref{E:ft-diff-bound} then gives
\[\frac{t}{2}\int_{-\frac{2}{t}}^{\frac{2}{t}}(1-\psi_n(\xi))d\xi=\frac{t}{2}\int_{-\frac{2}{t}}^{\frac{2}{t}}(1-\varphi(\xi))d\xi+\frac{t}{2}\int_{-\frac{2}{t}}^{\frac{2}{t}}(\varphi(\xi)-\psi_n(\xi))d\xi\le \frac{c}{t^2}\left(1+\frac{1}{\sqrt{n}}\right).\]

\medskip

For part \ref{P:dos-trunc-error}, note that by construction, $|f(x)-f_R(x)|\le1$; moreover, $f(x)=f_R(x)$ for $|x|\le R$, so that
\begin{equation*}\begin{split}
\left|\int f d \mu_n^{DOS} - \int f_R d \mu_n^{DOS}\right|& \le 
\int_{|x|>R} d  \mu_n^{DOS}.
\end{split}\end{equation*}
Part \ref{P:dos-tails} is now immediate from part \ref{P:dos-trunc-error}.

\end{proof}

Now let $f:\R\to\R$ have $\|f\|_{BL}\le 1$ and  $supp(f)\subseteq[-2R,2R]$. The next step in the proof of Theorem \ref{T:W_1_to_gamma} is to approximate $f$ by 
\[f_\lambda:=f*K_\lambda,\]
 where $K_\lambda$ is the F\'ejer kernel
\[K_\lambda(x)=\frac{1}{2\pi}\int_{-\lambda}^\lambda\left(1-\frac{|\xi|}{\lambda}\right)e^{i\xi x}d\xi=\frac{\lambda}{2\pi}\left(\frac{\sin(\lambda x/2)}{\lambda x/2}\right)^2.\]
For $f$ as above, one can approximate in the supremum norm, as follows.

\begin{lemma}\label{T:Fejer-approx}
Let $f:\R\to\R$ have $\|f\|_{BL}\le 1$ and $supp(f)\subseteq[-2R,2R]$.  Then
\[|f(x)-f_\lambda(x)|\le \frac{8\log(\lambda)+8\log(2R)+6}{\pi\lambda}.\]
\end{lemma}

\begin{proof}
By definition of $f_\lambda$ (using the second form of the F\'ejer kernel) and the fact that $\int_{-\infty}^\infty K_\lambda(y)dy=1$, 
\begin{equation*}\begin{split}|f(x)-f_\lambda(x)|&
=\frac{1}{\pi}\left|\int_\infty^\infty \left(f(x)-f\left(x-\frac{2y}{\lambda}\right)\right)\left(\frac{\sin(y)}{y}\right)^2dy\right|.\end{split}\end{equation*}
Now, if $x\in[-2R,2R]$, then using the fact that $f$ is supported on $[-2R,2R]$ and is 1-Lipschitz yields
\begin{equation*}\begin{split}
\pi|f(x)-f_\lambda(x)|&\le|f(x)|\int_{-\infty}^{\frac{\lambda}{2}(x-2R)}\left|\frac{\sin(y)}{y}\right|^2dy+\int_{\frac{\lambda}{2}(x-2R)}^{-1}\left|\frac{2\sin^2(y)}{\lambda y}\right|dy+\int_{-1}^1\left|\frac{2\sin^2(y)}{\lambda y}\right|dy\\&\qquad\qquad\qquad+\int_1^{\frac{\lambda}{2}(2R-x)}\left|\frac{2\sin^2(y)}{\lambda y}\right|dy+|f(x)|\int_{\frac{\lambda}{2}(2R-x)}^\infty \left|\frac{\sin(y)}{y}\right|^2dy\\&=:I+II+III+IV+V.
\end{split}\end{equation*}
Since $f(2R)=0$, 
\begin{equation*}\begin{split}
I\le|f(x)-f(2R)|\int_{-\infty}^{\frac{\lambda}{2}(x-2R)}\frac{1}{y^2}dy=\frac{2|f(x)-f(2R)|}{\lambda(2R-x)}\le\frac{2}{\lambda};
\end{split}\end{equation*}
$V$ is handled the same way. Next, since $|x|\le 2R$,
\[II\le \int_{\frac{\lambda}{2}(x-2R)}^{-1}\frac{2}{\lambda |y|}dy=\frac{2}{\lambda}\log\left(\frac{\lambda}{2}(2R-x)\right)\le\frac{4(\log(\lambda)+\log(2R))}{\lambda};\]
$IV$ is the same.
Finally, using the bound $\left|\frac{\sin(y)}{y}\right|\le 1$ gives
\[III\le \int_{-1}^1\frac{2|y|}{\lambda}dy=\frac{2}{\lambda}.\]

\smallskip

If $x>2R$, then by the concavity of the logarithm,
\begin{equation*}\begin{split}\pi|f(x)-f_\lambda(x)|&\le\int_{\frac{\lambda}{2}(x-2R)}^{\frac{\lambda}{2}(x+2R)}\left|\frac{2\sin^2(y)}{\lambda y}\right|dy\\&\le \frac{2}{\lambda}\left[\log\left(\frac{\lambda}{2}(x+2R)\right)-\log\left(\frac{\lambda}{2}(x-2R)\right)\right]\le\frac{2}{\lambda}\left(\frac{4R}{x}\right)\le\frac{4}{\lambda}.\end{split}\end{equation*}
The case $x<-2R$ is the same.
\end{proof}

The following technical lemma is needed in order to compare $\int f_\lambda d\mu^{DOS}_n$ to $\int f_\lambda d\gamma$.

\begin{lemma}\label{T:g_lambda_L_1}
Let $g:\R\to\R$ be such that $|g(\xi)|\le \min\left\{\frac{2}{|\xi|}, \frac{C|\xi|}{\sqrt{n}}\right\}.$
Then 
\[\int_{-R}^R |g_\lambda(x)|dx\le cR^2\left(\frac{1}{\sqrt{n}}+\frac{1}{\lambda}\right)+\frac{c'\log(\lambda)}{\lambda}.\]
\end{lemma}

\begin{proof}
Recall that $g_\lambda(x)=g*K_\lambda(x)$, so that
\begin{equation}\begin{split}\label{E:int_g_lambda}
\int_{-R}^R |g_\lambda(x)|dx&=\frac{1}{\pi}\int_{-R}^R\left|\int_{-\infty}^\infty g\left(x-\frac{2w}{\lambda}\right)\left(\frac{\sin(w)}{w}\right)^2dw\right|dx\\&\le\frac{2}{\pi}\int_{0}^R\int_{-\infty}^\infty \min\left\{\frac{2}{\left|x-\frac{2w}{\lambda}\right|}, \frac{C\left|x-\frac{2w}{\lambda}\right|}{\sqrt{n}}\right\}\left(\frac{\sin(w)}{w}\right)^2dwdx.\end{split}\end{equation}

Assume that $\lambda$ will be chosen with $\lambda>2$.
For $x>0$,
\begin{equation*}\begin{split}
\int_{0}^\infty &\min\left\{\frac{2}{\left|x-\frac{2w}{\lambda}\right|}, \frac{C\left|x-\frac{2w}{\lambda}\right|}{\sqrt{n}}\right\}\left(\frac{\sin(w)}{w}\right)^2dw\\&\le
\int_0^1 \frac{C\left|x-\frac{2w}{\lambda}\right|}{\sqrt{n}}dw+\int_1^{\frac{\lambda}{2}(x+1)}\frac{C\left|x-\frac{2w}{\lambda}\right|}{w^2\sqrt{n}}dw+\int_{\frac{\lambda}{2}(x+1)}^\infty\frac{2}{w^2\left|x-\frac{2w}{\lambda}\right|}dw.
\end{split}\end{equation*}

Now, the first term is trivially bounded by $\frac{C(x+1)}{\sqrt{n}}$.  For the second term,
\begin{equation*}\begin{split}
\int_1^{\frac{\lambda}{2}(1+x)}&\frac{C\left|x-\frac{2w}{\lambda}\right|}{w^2\sqrt{n}}dw\\&\le\frac{C}{\sqrt{n}}\int_1^{\frac{\lambda}{2}(1+x)}\frac{x+\frac{2w}{\lambda}}{w^2}dw=\frac{C}{\sqrt{n}}\left[x\left(1-\frac{2}{\lambda(1+x)}\right)+\frac{2}{\lambda}\log\left(\frac{\lambda}{2}(1+x)\right)\right].
\end{split}\end{equation*}
For the final term,
\begin{equation*}\begin{split}
\int_{\frac{\lambda}{2}(1+x)}^\infty \frac{2}{w^2 \left|x-\frac{2w}{\lambda}\right|}dw&
=\frac{4}{\lambda x^2}\int_{1+\frac{1}{x}}^\infty \frac{1}{t^2(t-1)}dt=\frac{4}{\lambda x^2}\left(\log(1+x)-\frac{x}{1+x}\right),
\end{split}\end{equation*}
and so 
\begin{equation*}\begin{split}
\int_{0}^{R}\int_{0}^\infty &\min\left\{\frac{2}{\left|x-\frac{2w}{\lambda}\right|}, \frac{C\left|x-\frac{2w}{\lambda}\right|}{\sqrt{n}}\right\}\left(\frac{\sin(w)}{w}\right)^2dwdx\\&\le\int_0^{R}\left(\frac{C}{\sqrt{n}}\left[2x+1+\frac{2}{\lambda}\log\left(\frac{\lambda}{2}(1+x)\right)-\frac{2x}{\lambda(1+x)}\right]+\frac{4}{\lambda x^2}\left(\log(1+x)-\frac{x}{1+x}\right)\right)dx\\&\le cR^2\left(\frac{1}{\sqrt{n}}+\frac{1}{\lambda}\right)+\frac{c\log(\lambda)}{\lambda}.
\end{split}\end{equation*}

Similarly, for $x<0$,
\begin{equation*}\begin{split}
\int_{-\infty}^0 &\min\left\{\frac{2}{\left|x-\frac{2w}{\lambda}\right|}, \frac{C\left|x-\frac{2w}{\lambda}\right|}{\sqrt{n}}\right\}\left(\frac{\sin(w)}{w}\right)^2dw\\&\le
\int_0^{1} \frac{C\left(x+\frac{2w}{\lambda}\right)}{\sqrt{n}}dw+\int_1^{\frac{\lambda}{2}(1+x)}\frac{C\left(x+\frac{2w}{\lambda}\right)}{w^2\sqrt{n}}dw+\int_{\frac{\lambda}{2}(1+x)}^\infty\frac{2}{w^2\left(x+\frac{2w}{\lambda}\right)}dw\\&\le \frac{C(x+1)}{\sqrt{n}}+\frac{C}{\sqrt{n}}\left[x\left(1-\frac{2}{\lambda(1+x)}\right)+\frac{2}{\lambda}\log\left(\frac{\lambda}{2}(1+x)\right)\right]+\frac{4}{\lambda x^2}\left[\frac{x}{x+1}-\log\left(\frac{2x+1}{x+1}\right)\right],
\end{split}\end{equation*}
and so
\begin{equation*}\begin{split}
\int_{0}^{R}\int_{-\infty}^0 &\min\left\{\frac{2}{\left|x-\frac{2w}{\lambda}\right|}, \frac{C\left|x-\frac{2w}{\lambda}\right|}{\sqrt{n}}\right\}\left(\frac{\sin(w)}{w}\right)^2dwdx\\&\le\int_0^{R}\left(\frac{C}{\sqrt{n}}\left[2x+1+\frac{2}{\lambda}\log\left(\frac{\lambda}{2}(1+x)\right)-\frac{2x}{\lambda(1+x)}\right]+\frac{4}{\lambda x^2}\left[\frac{x}{x+1}-\log\left(\frac{2x+1}{x+1}\right)\right]\right)dx\\&\le cR^2\left(\frac{1}{\sqrt{n}}+\frac{1}{\lambda}\right)+\frac{c'\log(\lambda)}{\lambda}.
\end{split}\end{equation*}

\end{proof}

\begin{prop}\label{T:gaussian-dos-f_lambda}
Let $f:\R\to\R$ have $\|f\|_{BL}\le 1$ and  $supp(f)\subseteq[-2R,2R]$.  Then
\[\left|\int f_\lambda d\mu^{DOS}_n-\int f_\lambda d\gamma\right|\le cR^2\left(\frac{1}{\sqrt{n}}+\frac{1}{\lambda}\right)+\frac{c'\log(\lambda)}{\lambda}.\]
\end{prop}

\begin{proof}
Let $\psi_n(\xi)$ denote the characteristic function of $\mu^{DOS}_n$ and $\varphi(\xi)$ the characteristic function of the standard Gaussian distribution $\gamma$.
Then by the various definitions,
\begin{equation*}\begin{split}
&\left|\int f_\lambda d\mu_n^{DOS}-\int f_\lambda d\gamma\right|\\&\qquad=\frac{1}{2\pi}\left|\int_{-\infty}^\infty \int_{-\lambda}^\lambda\hat{f}(\xi)e^{ix\xi}\left(1-\frac{|\xi|}{\lambda}\right)d\xi d\mu_n^{DOS}(x)-\int_{-\infty}^\infty \int_{-\lambda}^\lambda\hat{f}(\xi)e^{ix\xi}\left(1-\frac{|\xi|}{\lambda}\right)d\xi d\gamma(x)\right|\\&\qquad=\frac{1}{2\pi}\left|\int_{-\lambda}^\lambda\hat{f}(\xi)\left(1-\frac{|\xi|}{\lambda}\right)\left(\int_{-\infty}^\infty e^{ix\xi} d\mu_n^{DOS}(x)-\int_{-\infty}^\infty e^{ix\xi} d\gamma(x)\right)d\xi\right|\\&\qquad=\frac{1}{2\pi}\left|\int_{-\lambda}^\lambda\hat{f}(\xi)\left(1-\frac{|\xi|}{\lambda}\right)\Big(\psi_n(-\xi)-\varphi(-\xi)\Big)d\xi\right|\\&\qquad=\frac{1}{2\pi}\left|\int_{-\lambda}^\lambda \left(\int_{-\infty}^\infty f(x)e^{-ix\xi}dx\right)\left(1-\frac{|\xi|}{\lambda}\right)\Big(\psi_n(-\xi)-\varphi(-\xi)\Big)d\xi \right|\\&\qquad=\frac{1}{2\pi}\left|\int_{-\lambda}^\lambda \left(\int_{-\infty}^\infty \frac{f'(x)}{i\xi}e^{-ix\xi}dx\right)\left(1-\frac{|\xi|}{\lambda}\right)\Big(\psi_n(-\xi)-\varphi(-\xi)\Big)d\xi \right|\\&\qquad=\frac{1}{2\pi}\left|\int_{-2R}^{2R}f'(x)\int_{-\lambda}^\lambda \left(\frac{\psi_n(\xi)-\varphi(\xi)}{\xi}\right)\left(1-\frac{|\xi|}{\lambda}\right) e^{ix\xi} d\xi dx \right|\\&\qquad=\left|\int_{-2R}^{2R}f'(x) \left(\frac{\psi_n(\xi)-\varphi(\xi)}{\xi}\right)_\lambda(x) dx \right|,
\end{split}\end{equation*}
where the third to last line follows by integration by parts and we have used that $f$ is supported on $[-2R,2R]$ in the last two lines.

We can now apply the result of Lemma \ref{T:g_lambda_L_1} and the fact the $\|f\|_{BL}\le 1$ to obtain the conclusion.

\end{proof}

We are now ready to give the proof of Theorem \ref{T:W_1_to_gamma}
\begin{proof}[Proof of Theorem \ref{T:W_1_to_gamma}]
Let $f:\R\to\R$ have $\|f\|_{BL}\le 1$.  Then by Lemma \ref{T:truncation} (and its much stronger counterpart for $\gamma$),
\begin{equation*}
\sup_{|f|_L\le 1}\left|\int fd\mu_n^{DOS}-\int fd\gamma\right|\le \frac{c}{R^2}+\sup_{\substack{|f|_L\le 1,\\ supp(f)\subseteq[-2R,2R]}}\left|\int fd\mu_n^{DOS}-\int fd\gamma\right|.
\end{equation*}
By Lemma \ref{T:Fejer-approx}, for $f$ with $\|f\|_{BL}\le 1$ and  support in $[-2R,2R]$,
\[\sup_x|f(x)-f_\lambda(x)|\le \frac{8\log(\lambda)+8\log(2R)+6}{\pi\lambda}, \]
and so 
\begin{equation*}\begin{split}
\sup_{\substack{|f|_L\le 1,\\ supp(f)\subseteq[-2R,2R]}}&\left|\int fd\mu_n^{DOS}-\int fd\gamma\right|\\&\le \sup_{\substack{|f|_L\le 1,\\ supp(f)\subseteq[-2R,2R]}}\left|\int f_\lambda d\mu_n^{DOS}-\int f_\lambda d\gamma\right|+\frac{16\log(\lambda)+16\log(2R)+12}{\pi\lambda}.\end{split}
\end{equation*}
Applying Proposition \ref{T:gaussian-dos-f_lambda} now gives
\[d_{BL}(\mu_n^{DOS},\gamma)\le \frac{c}{R^2}+\frac{16\log(\lambda)+16\log(2R)+12}{\pi\lambda}+cR^2\left(\frac{1}{\sqrt{n}}+\frac{1}{\lambda}\right)+\frac{c'\log(\lambda)}{\lambda}.\]
Choosing $\lambda =n$ and $R=n^{1/8}$ completes the proof.
\end{proof}

\section{Concentration and average distance to average}\label{S:avg-to-avg} 
A crucial underpinning of the remainder of the proof is the following concentration of measure property of a Gaussian random vector.

\begin{prop}[See, e.g., Ch.\ 1 of \cite{Led}]
\label{T:Gaussian-concentration}
Let $(Z_k)_{1\le k\le n}$ be a standard $n$-dimensional Gaussian random vector, 
and let $F: \mathbb{R}^n \rightarrow \mathbb{R}$ be Lipschitz with Lipschitz constant $L$.  There are universal constants $C,c$ such that
$$\mathbb{P}[ | F(Z_1,\ldots,Z_n) - \mathbb{E} F(Z_1,\ldots,Z_n) | > t] \le Ce^{-ct^2/L^2}.$$
\end{prop}

The concentration phenomenon is key both in proving the concentration of the bounded-Lipschitz distance from $\mu_n$ to a fixed reference measure, and in estimating $\E c_{BL}(\mu_n,\mu_n^{DOS})$.  
The following lemma gives the necessary Lipschitz estimates for this approach.
\begin{lemma}
\label{T:dist-is-Lipschitz}
Let ${\bf x}=\{x_{a,b,j}\}\in\R^{9n}$ (with, say, lexicographic ordering).  Define $H_n({\bf x})$ by 
\[H_n({\bf x}):=\frac{1}{3\sqrt{n}}\sum_{a,b=1}^3\sum_{j=1}^nx_{a,b,j}\sigma^{(a)}_j\sigma^{(b)}_{j+1},\]
and let $\mu_n$ be the spectral measure of $H_n({\bf x})$.  Let $f:\R\to\R$ have $\|f\|_{BL}\le 1$.  Then
\begin{enumerate}[label=(\alph*)]
\item the map \label{P:integral-is-Lipschitz}
$${\bf x}\mapsto \int f d\mu_n$$
 is $\frac{1}{3 \sqrt {n}}$-Lipschitz,  
 and
\item for any probability measure $\rho$ on $\R$, the map \label{P:dist-is-Lipschitz}
 $${\bf x}\mapsto d_{BL}( \mu_n ,\rho)$$
  is $\frac{1}{3 \sqrt {n}}$-Lipschitz.
\end{enumerate}
\end{lemma}
\begin{proof}

First consider the map ${\bf x}\mapsto H_n({\bf x})$, and equip the space of $2^n\times 2^n$ symmetric matrices with the Hilbert-Schmidt norm: 
\[\|A\|_{H.S.}:=\tr(AA^T)=\tr(A^2).\]  For ${\bf x},{\bf x}'\in\R^{9n}$, write $H_n:=H_n({\bf x})$ and $H_n':=H_n({\bf x}')$.  Then
 \begin{equation*}\begin{split}
\| H_n - H'_n \|_{H.S.}^2 &= \tr\big[(H_n-H'_n)^2\big]\\
&=\frac{1}{9n}\sum_{j,k=1}^n  \sum_{a,b,c,d=1}^3 
 (x_{a,b,j}-x'_{a,b,j}) (x_{c,d,k}-x'_{c,d,k}) 
 \tr(\sigma_j^{(a)} \sigma_{j+1}^{(b)}\sigma_k^{(c)} \sigma_{k+1}^{(d)}).\end{split}\end{equation*}
Recall that $\tr(A\otimes B)=\tr(A)\tr(B)$, and that $\tr(\sigma^{(a)})=0$ for each of the Pauli matrices.  If $j\neq k$, then the matrix $\sigma_j^{(a)} \sigma_{j+1}^{(b)}\sigma_k^{(c)} \sigma_{k+1}^{(d)}$ is a tensor product, at least two of whose factors are Pauli matrces; that is, if $j\neq k$, then 
\[\tr(\sigma_j^{(a)} \sigma_{j+1}^{(b)}\sigma_k^{(c)} \sigma_{k+1}^{(d)})=0.\]

 If $j=k$,
 then 
\[\sigma_j^{(a)} \sigma_{j+1}^{(b)}\sigma_k^{(c)} \sigma_{k+1}^{(d)}=I_2^{\otimes (j-1)}\otimes \sigma^{(a)}\sigma^{(c)}\otimes\sigma^{(b)}\sigma^{(d)}\otimes I_2^{\otimes(n-j-1)},\]
and thus
\[\tr(\sigma_j^{(a)} \sigma_{j+1}^{(b)}\sigma_k^{(c)} \sigma_{k+1}^{(d)})=\begin{cases}2^n,&a=c, b=d;\\0,&otherwise.\end{cases}\] 
It follows that
 $$\|H_n-H'_n \|_{H.S.}= \sqrt{\frac{2^n}{9n} \sum_{j=1}^n \sum_{a,b=1}^3 
 (x_{a,b,j}-x'_{a,b,j})^2}= \frac{2^{n/2}}{3\sqrt{n}} \| {\bf x}- {\bf
 x}' \|,$$
 and so the map ${\bf x}\mapsto H_n$ is $\frac{2^{n/2}}{3\sqrt{n}}$-Lipschitz.

Now consider the map $H_n \mapsto \int f d\mu_{n}$.
 By definition, $\int f d\mu_n= \frac{1}{2^n} \sum_{j=1}^{2^n} f(\lambda_j)$ so
 \[ 
 \left|\int f d\mu_n - \int f d\mu_n' \right| = 2^{-n} \left| \sum_{j=1}^{2^n} f(\lambda_j)- f(\lambda'_j) \right| \le 
  2^{-n} \sum_{j=1}^{2^n} |\lambda_j- \lambda'_j|\le 2^{-n/2} \sqrt{ \sum_{j=1}^{2^n} |\lambda_j- \lambda'_j|^2},\]
making use of the fact that $f$ is 1-Lipschitz.

The Hoffman-Wielandt inequality (see, e.g., \cite[Theorem VI.4.1]{Bha}) gives that
\[
2^{-n/2} \sqrt{ \sum_{j=1}^{2^n} |\lambda_j- \lambda'_j|^2} \le
2^{-n/2} \|  H_n - H'_n  \|_{H.S.}.
\]
and so the map $H_n\mapsto\int fd\mu_n$ is $2^{-n/2}$-Lipschitz; this completes the proof of part \ref{P:integral-is-Lipschitz}.

For part \ref{P:dist-is-Lipschitz}, first note that by the triangle inequality for $d_{BL}$,
\[\left|d_{BL}(\mu_n,\rho)-c_{BL}(\mu_n',\rho)\right|\le c_{BL}(\mu_n,\mu_n').\]
Define a coupling $\pi$ of $\mu_n$ and $\mu_n'$ by 
\[\pi := \frac{1}{2^n} \sum_{j=1}^{2^n} \delta_{(\lambda_j, \lambda_j')},\] where 
$\lambda_j$ and $\lambda_j'$ are ordered eigenvalues of $H_n$ and $H'_n$ respectively.
Then by the Kantorovich--Rubenstein theorem,
\[d_{BL}(\mu_n,\mu_n')\le W_1(\mu_n,\mu_n')\le\int |x-y|d\pi(x,y)=\frac{1}{2^n}\sum_j|\lambda_j-\lambda_j'|.\]
 Applying the Cauchy-Schwarz inequality and the Hoffman-Wielandt inequality exactly as before gives that the map $H_n\mapsto d_{BL}(\mu_n,\rho)$ is $2^{-n/2}$-Lipschitz; together with the Lipschitz estimate for ${\bf x}\mapsto H_n({\bf x})$ given above, this completes the proof.
\end{proof}

It is thus immediate from Proposition \ref{T:Gaussian-concentration} and Lemma \ref{T:dist-is-Lipschitz} that if $\rho$ is any probability measure,
\begin{equation}\label{E:conc-of-distance}
\P\left[\left|d_{BL}(\mu_n,\rho)-\E d_{BL}(\mu_n,\rho)\right|>t\right]\le Ce^{-cnt^2},\end{equation}
and so part \ref{P:dist-conc-gauss} of Theorem
\ref{T:random-dist-to-gauss} follows immediately from part \ref{P:exp-dist-to-gauss}.

\bigskip

To prove part \ref{P:exp-dist-to-gauss}, recall that in the previous section it was shown that
\[d_{BL}(\mu_n^{DOS},\gamma)\le\frac{C}{n^{1/4}}.\]
It thus suffices by the triangle inequality for $d_{BL}$ to show that
\begin{equation}\label{E:avg-to-avg}\E d_{BL}(\mu_n,\mu_n^{DOS})=\sup_{\|f\|_{BL}\le 1}\left(\int fd\mu_n-\int fd\mu_n^{DOS}\right)\le\frac{C}{n^{1/6}}.\end{equation}
Observe first that $f$ with $\|f\|_{BL}\le 1$, if
\[X_f := \int f d \mu_{n} - \int f d \mu_{n}^{DOS},\]
then $\E X_f=0$ and by Proposition \ref{T:Gaussian-concentration} and Lemma \ref{T:dist-is-Lipschitz},
\[
\mathbb{P}[|X_f | > t] \le Ce^{-c n t^2}.\]
More generally,
\begin{equation}\label{E:sub-Gauss}
\P\left[\left|X_f-X_g\right|>t\right]=\P\left[\left|X_{f-g}\right|>t\right]\le Ce^{-\frac{cnt^2}{\|f-g\|_{BL}^2}};
\end{equation}
that is, the process $\{X_f\}$ indexed by functions $f:\R\to\R$ with $\|f\|_{BL}\le 1$ and  $f(0)=0$ is a sub-Gaussian stochastic process, with respect to the distance $d(f,g)=\frac{\|f-g\|_{BL}}{\sqrt{n}}.$  

\smallskip

The idea at this point is to use Dudley's entropy bound to estimate the expected supremum of this process, but to do this successfully, a series of approximations must be made first to reduce the size of the (currently infinite-dimensional) indexing set of the process.  The first step is a somewhat more sophisticated truncation argument than the one which appeared in Section \ref{S:DOS-to-gamma}, which allows us to assume that our test functions are finitely supported.  

Let $\|A\|_{op}$ denote the $\ell_2\to\ell_2$ operator norm of a matrix $A$, and observe that the map ${\bf x}\mapsto \|H_n\|_{op}$ is 1-Lipschitz:
\begin{equation*}\begin{split}
| \|  H_n  \|_{op} - \|  H'_n  \|_{op}|  &\le \|  H_n -  H'_n  \|_{op}
\\&= \frac{1}{3\sqrt{n}}
\| \sum_{j=1}^n \sum_{a,b=1}^3 (x_{a,b,j} - x'_{a,b,j}) \sigma_j^{(a)} \sigma_{j+1}^{(b)} \|_{op} \\&\le
\frac{1}{3\sqrt{n}}\sum_{j=1}^n \sum_{a,b=1}^3 |(x_{a,b,j} - x'_{a,b,j})| \| \sigma_j^{(a)} \sigma_{j+1}^{(b)} \|_{op}.\end{split}\end{equation*}

Now, $\| A \otimes B \|_{op} = \| A \|_{op} \|B \|_{op}$ 
and all of the Pauli matrices have operator norm $1$, so
$ \| \sigma_j^{(a)} \sigma_{j+1}^{(b)} \|_{op} = 1$ for all $a,b,j$.  An application of 
the Cauchy-Schwarz inequality thus gives that
\begin{equation*}\begin{split}
\big| \|  H_n  \|_{op} - \|  H'_n  \|_{op}\big|  &  \le
 \sqrt{ \sum_{j=1}^n \sum_{a,b=1}^3 |x_{a,b,j} - x'_{a,b,j}|^2 }=\|{\bf x}-{\bf x}'\|.\end{split}\end{equation*}

It thus follows from concentration of measure that 
\[\mathbb{P}[ | \|  H_n  \|_{op} - \mathbb{E}  \|  H_n  \|_{op} | > t] \le Ce^{-ct^2}.\]
Since
\[\E \|  H_n  \|_{op}\le\frac{1}{3\sqrt{n}}\sum_{a,b=1}^3\sum_{j=1}^n\E|Z_{a,b,j}|=\frac{3\sqrt{2n}}{\sqrt{\pi}}.\]
one has in particular that if $t\ge C\sqrt{n}$, then 
\[\mathbb{P}[ \|  H_n  \|_{op}  > t] \le Ce^{-ct^2}.\]
One can interpret this statement as saying that if $R\sim\sqrt{n}$ it is extremely unlikely that $H_n$ will have any eigenvalues outside $[-R,R]$, and so truncation of test functions to that interval should not result in much loss.

More specifically, if $f_R$ is the truncation of $f$ to $[-2R,2R]$ given in Equation \eqref{E:truncation} of Section \ref{S:DOS-to-gamma}, then 
\[\E\left|\int(f-f_R)d\mu_n\right|\le C\mu_n^{DOS}\left(\left\{x:|x|>R\right\}\right). \]
Since
\[\mu_n^{DOS}\left(\left\{x:|x|>R\right\}\right)=\frac{1}{2^n}\E\left|\big\{j:|\lambda_j|>R\big\}\right|\le\P\left[\|H_n\|_{op}>R\right],\]
it follows that 
\[\E\left|\int(f-f_R)d\mu_n\right|\le Ce^{-cnR^2}.\]
The indexing space of the process $\{X_f\}$ may thus be safely reduced to those $f$ supported on $[-2R,2R]$ with $R$ of order $\sqrt{n}$, with an error which is exponentially small in $n$.

\medskip

The next step in reducing the indexing space is to approximate bounded Lipschitz test functions by piecewise linear ones.  Given $f:\R\to\R$ with $\|f\|_{BL}\le 1$ and $supp(f)\subseteq[-2R,2R]$, consider the piecewise linear function $g:\R\to\R$ defined so that $supp(g)\subseteq[-2R,2R]$, and $g(x)=f(x)$ at each $x$ of the form $-2R+\frac{4Rk}{m}$, for $0\le k\le  m$.
Because $f$ is $1$-Lipschitz,
\[ \|f-g\|_{\infty} \le \frac{2R}{m},\]
and so
\[|  X_f - X_g  | \le \frac{4 R}{m}.\]

It follows that
\begin{equation}\label{E:all-but-dudley}
\E d_{BL}(\mu_{n},\mu_n^{DOS}) \le \E\sup_{g\in\mathcal{G}} X_g  + \frac{4R}{m}+Ce^{-cR^2},\end{equation}
where the supremum is taken over the class $\mathcal{G}$ of functions $g:\R\to\R$ satisfying
\begin{itemize}
\item $g(0)=0$;
\item $\|g\|_{BL}\le 1$;
\item $supp(g)\subseteq[-2R,2R]$;
\item $g$ is linear on intervals of the form $\left[-2R+\frac{4Rk}{m},-2R+\frac{4R(k+1)}{m}\right]$.
\end{itemize}

With the reduction to $\mathcal{G}$ as the indexing space of our stochastic process, it is now possible to apply Dudley's entropy bound (see, e.g., the introduction of \cite{Tal}):

\begin{prop}[Dudley]\label{T:dudley}
Let $\{ Y_x : x \in M \}$ be a centered subgaussian stochastic process indexed by the metric space $(M,d)$.
Then
$$ \mathbb{E} \sup_{x \in M} |Y_x| \le K \int_0^{\infty} \sqrt{\log N(M, d, \epsilon)} d\epsilon, $$
where $N(M,d,\epsilon)$ denotes the number of $\epsilon$-balls (with respect to the metric $d$) needed to cover $M$, and $K > 0$ depends only on the constants of the sub-Gaussian increment condition.
\end{prop}

Let $\mathcal{G}$ denote the index set described above.  Applying Proposition \ref{T:dudley} to $\{X_g\}_{g\in\mathcal{G}}$ gives that 
\[\E\sup_{g\in\mathcal{G}}X_g\le K\int_0^\infty\sqrt{\log\left[N\left(\mathcal{G},\frac{|\cdot|_L}{\sqrt{n}},\epsilon\right)\right]}d\epsilon=\frac{K}{\sqrt{n}}\int_0^\infty\sqrt{\log\left[N\left(\mathcal{G},|\cdot|_L,\epsilon\right)\right]}d\epsilon.\]
Since $\left(\mathcal{G},|\cdot|_L\right)$ is just an $(m+1)$-dimensional normed space, standard volumetric estimates (see \cite[Lemma 2.6]{MS}) give that 
\[N\left(\mathcal{G},|\cdot|_L,\epsilon\right)\le\left(\frac{3}{\epsilon}\right)^{m+1},\]
so that together with \eqref{E:all-but-dudley}, we have that
 \[\E W_1(\mu_{H_n},\mu_n^{DOS}) \le K\sqrt{\frac{m}{n}}  + \frac{2R}{m}+Ce^{-cR^2}.\]
Choosing $R$ of order $\sqrt{n}$ and $m=n^{2/3}$ completes the proof of \eqref{E:avg-to-avg}, and thus the proof of part \ref{P:exp-dist-to-gauss} of Theorem \ref{T:random-dist-to-gauss}.

\section{The spherical model}\label{S:sphere}
As discussed in the introduction, all of the random matrix models of
quantum spin chains considered so far involve independent
coefficients.  A perhaps more geometrically natural alternative is to consider the random matrix 
\begin{equation}\label{E:H_n-sphere}H_n:=\sum_{j=1}^n\sum_{a,b=1}^3x_{a,b,j}\sigma^{(a)}_j\sigma^{(b)}_{j+1},\end{equation}
where the $\sigma^{(a)}_j$ are as before, but the vector of
coefficients ${\bf x}=\{x_{a,b,j}\}_{\substack{1\le a,b\le 3\\1\le j\le n}}$
is chosen uniformly from the unit sphere in $\R^{9n}$.  While this
model introduces dependence among the coefficients, it is still
possible to prove the almost sure convergence of the empirical
spectral measure $\mu_n$ of $H_n$ to the standard Gaussian
distribution, albeit without a specific rate.
\begin{thm}\label{T:sphere-as-convergence}
For each $n\ge 1$, let $\mu_n$ be the spectral measure of the random matrix $H_n$ defined
as in \eqref{E:H_n-sphere}.  Then almost surely, the sequence
$\{\mu_n\}_{n\in\N}$ tends weakly to the standard Gaussian
distribution, as $n\to\infty$.
\end{thm}
It should be noted that the statement above implicitly assumes some
joint distribution of the coefficient vectors, but the theorem is true
independent of what that joint distribution is.

That the earlier results for i.i.d.\ coefficients can be extended to
this dependent setting relies on two important
properties of uniform random vectors on the sphere.  The first is that
explicit computations are still at least somewhat feasible.  The second is
that the concentration of measure phenomenon which was strongly used
in the Gaussian case holds for random vectors on the sphere as well,
as follows.
\begin{lemma}[L\'evy's lemma; see [Ch.\ 1 of \cite{Led}]\label{T:Levy}
Let $(x_k)_{1\le k\le n}$ be a random vector, uniformly distributed on
$\mathbb{S}^{n-1}$, 
and let $F: \mathbb{R}^n \rightarrow \mathbb{R}$ be Lipschitz with Lipschitz constant $L$.  There are universal constants $C,c$ such that
$$\mathbb{P}[ | F(x_1,\ldots,x_n) - \mathbb{E} F(x_1,\ldots,x_n) | > t] \le Ce^{-cnt^2/L^2}.$$
\end{lemma}

The following modified version of Theorem \ref{T:random-dist-to-gauss} holds
for the spherical model; here we compare $\mu_n$ to
$\mu_n^{DOS}$ rather than to the Gaussian distribution.  
\begin{thm}\label{T:random-dist-to-gauss-sphere}
Let $\mu_n$ be the spectral measure of $H_n$ and let
$\mu_n^{DOS}:=\E\mu_n$ be the density of states measure.  There are universal constants $C$, $C'$ and $c$ such that
\begin{enumerate}[label=(\alph*)]
\item \label{P:exp-dist-to-gauss}\(\displaystyle\E d_{BL}(\mu_n,\mu_n^{DOS})\le \frac{C}{n^{1/6}};\)
\item \label{P:dist-conc-gauss}\(\displaystyle\P\left[d_{BL}(\mu_n,\mu_n^{DOS})\ge\frac{C}{n^{1/6}}+t\right]\le Ce^{-cnt^2};\)
\end{enumerate}
and
\begin{enumerate}[resume,label=(\alph*)]
\item \label{P:as-dist-to-gauss}with probability 1, for all sufficiently large $n$,
\[d_{BL}(\mu_n,\mu_n^{DOS})\le \frac{C'}{n^{1/6}}.\]
\end{enumerate}

\end{thm}

All of the proofs in Section \ref{S:avg-to-avg} go through in exactly
the same way, using L\'evy's lemma in place of Proposition
\ref{T:Gaussian-concentration}.   (Note the difference in normalization:
Proposition \ref{T:Gaussian-concentration} is stated for a standard
Gaussian random vector, with expected length on the order of
$\sqrt{n}$, whereas L\'evy's lemma is stated for a random vector on
the unit sphere.)  The missing element in showing the almost sure
convergence of $\mu_n$ to the Gaussian distribution is the comparison
of $\mu_n^{DOS}$ to Gaussian, which in the case of i.i.d.\ Gaussian
coefficients in $H_n$ followed from the characteristic
function estimate proved in \cite{KLW}.  The following result gives an analog
of their result in for the spherical model, but without a similarly
good rate of convergence; this is the reason that we do not obtain an
almost sure convergence rate of $\mu_n$ in the spherical model.

\begin{prop}\label{T:doscf-gaussian-sphere}
Let $\psi_n(t)$ denote the characteristic function of the density of
states measure of $H_n$ as defined in Equation \eqref{E:H_n-sphere}.
Then for each $t\in\R$,
\[\psi_n(t)\xrightarrow{n\to\infty}e^{-\frac{t^2}{2}}.\]
\end{prop} 
To modify the approach in \cite{KLW} to prove Proposition \ref{T:doscf-gaussian-sphere}, we will need to calculate
expectations of certain functions over the unit sphere; the following
lemma gives explicit formulae.
\begin{lemma}[See \cite{Fol}]\label{T:sph-ints}
Let $P(x)=|x_1|^{\alpha_1}|x_2|^{\alpha_2}\cdots|x_n|^{\alpha_n}$.  
Then if $X$ is uniformly distributed on $S^{n-1}$, 
$$\E\big[P(X)\big]=\frac{\Gamma(\beta_1)\cdots\Gamma(\beta_n)\Gamma(\frac{n}{2})}{\Gamma(\beta_1+\cdots+\beta_n)\pi^{n/2}},$$
where $\beta_i=\frac{1}{2}(\alpha_i+1)$ for $1\le i\le n$ and 
$$\Gamma(t)=\int_0^\infty s^{t-1}e^{-s}ds=2\int_0^\infty r^{2t-1}e^{-r^2}dr.$$
\end{lemma}
The crucial technical ingredient for Proposition
\ref{T:doscf-gaussian-sphere} is the following.
\begin{lemma}\label{T:cosines-average}
Let $\{x_k\}_{1\le k\le N}$ be uniformly distributed on the unit
sphere in $\R^N$.  For each $t\in\R$,
\[\E\left[\prod_{k=1}^N\cos(tx_k)\right]\xrightarrow{N\to\infty}e^{-t^2}.\]
\end{lemma}
\begin{proof}We first show that it suffices to approximate the cosine with a
second-order Taylor expansion.  It follows from the L\'evy's lemma and the fact that $|\cos(tx_k)|\le 1$ that
\begin{equation}\begin{split}\label{E:restrict-to-cube}
\left|\E\prod_{k=1}^N\cos(tx_k)-\E\prod_{k=1}^N\cos(tx_k)\ind{\left|1-\frac{(tx_k)^2}{2}\right|\le
    1}\right|&=\left|\E\left[\prod_{k=1}^N\cos(tx_k)\ind{\bigcup_{k=1}^N\left\{\left|1-\frac{(tx_k)^2}{2}\right|>
      1\right\}}\right]\right|\\&
\le\P\left[\bigcup_{k=1}^N\left\{\left|1-\frac{(tx_k)^2}{2}\right|>
      1\right\}\right]\\&=\P\left[\bigcup_{k=1}^N\left\{\frac{(tx_k)^2}{2}>
      2\right\}\right]\le CNe^{-\frac{cN}{t^2}}.
\end{split}\end{equation}
From the trivial estimate that $|z_1\cdots z_m-w_1\cdots
w_m|\le\sum_{k=1}^m|z_k-w_k|$ if $|z_k|,|w_k|\le 1$ for all $k$, it
then follows that 
\begin{equation}\begin{split}\label{E:Taylor-replace}
\left|\E\prod_{k=1}^N\left\{\cos(tx_k)-1+\frac{(tx_k)^2}{2}\right\}\ind{\left|1-\frac{(tx_k)^2}{2}\right|\le
    1}\right|&\le\sum_{k=1}^N\E\left|\cos(tx_k)-1+\frac{(tx_k)^2}{2}\right|\ind{\left|1-\frac{(tx_k)^2}{2}\right|\le
    1}\\&\le \frac{N\E(tx_1)^4}{4!}\le\frac{t^4}{N\cdot 4!},
\end{split}\end{equation}
where the last line follows from Lemma \ref{T:sph-ints}.
Now, by the Cauchy-Schwarz inequality,
\begin{equation}\begin{split}\label{E:unrestrict-cs}
&\left|\E\prod_{k=1}^N\left\{1-\frac{(tx_k)^2}{2}\right\}\ind{\left|1-\frac{(tx_k)^2}{2}\right|\le
  1}-\E\prod_{k=1}^N\left\{1-\frac{(tx_k)^2}{2}\right\}\right|\\&\qquad\qquad\le\sqrt{ \E\prod_{k=1}^N\left(1-\frac{(tx_k)^2}{2}\right)^2}\sqrt{\P\left[\bigcup_{k=1}^N\left\{\left|1-\frac{(tx_k)^2}{2}\right|>
      1\right\}\right]}\\&\qquad\qquad\le Ce^{-\frac{cN}{t^2}}\sqrt{ \E\prod_{k=1}^N\left(1-(tx_k)^2+\frac{(tx_k)^4}{4}\right)}.
\end{split}\end{equation}
Expanding this last expression and using Lemma \ref{T:sph-ints} gives
that
\begin{equation}\begin{split}\label{E:unrestrict-bound}
\E\prod_{k=1}^N\left(1-(tx_k)^2+\frac{(tx_k)^4}{4}\right)&=\sum_{j=0}^N\binom{N}{j}\sum_{\ell=0}^{N-j}\binom{N-j}{\ell}\frac{(-t)^{2j}t^{4\ell}}{4^\ell}\E[x_1^2\cdots
x_j^2x_{j+1}^4\cdots x_{j+\ell}^4]\\&\le\sum_{j=0}^N\binom{N}{j}\left(\frac{-t^2}{N}\right)^j\sum_{\ell=0}^{N-j}\binom{N-j}{\ell}\left(\frac{3t^4}{4N^2}\right)^\ell\\&=\sum_{j=0}^N\binom{N}{j}\left(\frac{-t^2}{N}\right)^j\left(1+\frac{3t^4}{4N^2}\right)^{N-j}\\&=\left(1-\frac{t^2}{N}+\frac{3t^4}{4N^2}\right)^{N}.
\end{split}\end{equation}
Since this last expression is asymptotic to $e^{-t^2}$, it is in
particular 
bounded.  Combining equations \eqref{E:restrict-to-cube},
\eqref{E:Taylor-replace}, \eqref{E:unrestrict-cs}, and
\eqref{E:unrestrict-bound} gives that 
\begin{equation}\label{E:real-Taylor-replace}
\left|\E\prod_{k=1}^N\cos(tx_k)-\E\prod_{k=1}^N\left(1-\frac{(tx_k)^2}{2}\right)\right|\le \frac{Ct^4}{N},
\end{equation}
and it remains to analyze $\E\prod_{k=1}^N\left(1-\frac{(tx_k)^2}{2}\right)$.

\medskip

By Lemma \ref{T:sph-ints},
\begin{equation*}\begin{split}\mathbb{E}\Big[ \prod_{k=1}^N \Big(1-\frac{(t x_k)^2}{2}\Big)\Big]&=
 \sum_{k=0}^N {N \choose k} \Big(\frac{-t^2}{2} \Big)^k \frac{\Gamma(\frac{N}{2})}{2^k \Gamma(\frac{N}{2}+k)}\\&=
  \sum_{k=0}^N \frac{(\frac{-t^2}{2} )^k}{k!} \frac{N(N-1)\cdots(N-k+1)}{(N+2k-2)(N+2k-4)\cdots(N)}\\&=\sum_{k=0}^N \frac{(\frac{-t^2}{2} )^k}{k!} \left[\frac{\prod_{\ell=1}^{k-1}\left(1-\frac{\ell}{N}\right)}{\prod_{\ell=1}^{k-1}\left(1+\frac{2\ell}{N}\right)}\right].
\end{split}\end{equation*}
Clearly 
\[\frac{\prod_{\ell=1}^{k-1}\left(1-\frac{\ell}{N}\right)}{\prod_{\ell=1}^{k-1}\left(1+\frac{2\ell}{N}\right)}\le
1,\]
and applying Taylor's theorem to the logarithms gives that there is a
constant $C$ such that
\begin{equation*}\begin{split}
\frac{\prod_{\ell=1}^{k-1}\left(1-\frac{\ell}{N}\right)}{\prod_{\ell=1}^{k-1}\left(1+\frac{2\ell}{N}\right)}&=\exp\left(\sum_{\ell=1}^{k-1}\log\left(1-\frac{\ell}{n}\right)-\log\left(1+\frac{2\ell}{N}\right)\right)\\&\ge\exp\left(-\sum_{\ell=1}^{k-1}\left(\frac{3\ell}{N}+\frac{C\ell^2}{N^2}\right)\right)
\\&\ge\exp\left(-\frac{3k(k-1)}{2N}-\frac{Ck^3}{N^2}\right).\end{split}\end{equation*}
It follows that for any $m\le N$,
\begin{equation*}\begin{split}
\left|\mathbb{E}\prod_{k=1}^N \left(1-\frac{(t
      x_k)^2}{2}\right)-e^{-\frac{t^2}{2}}\right|&\le\sum_{k=0}^m
\frac{(\frac{t^2}{2} )^k}{k!}
\left[1-\exp\left(-\frac{3k(k-1)}{2N}-\frac{Ck^3}{N^2}\right)\right]+\sum_{k=m+1}^\infty
\frac{(\frac{t^2}{2} )^k}{k!}\\&\le
e^{\frac{t^2}{2}}\left(\frac{3m(m-1)}{2N}+\frac{Cm^3}{N^2}\right)+\frac{\left(\frac{t^2}{2}\right)^{m+1}}{(m+1)!}\\&\le \frac{Ce^{\frac{t^2}{2}}m^2}{N}+\frac{C}{\sqrt{m}}\left(\frac{et^2}{2(m+1)}\right)^{m+1}.
\end{split}\end{equation*}

Choosing, say,  $m=\left\lceil N^{1/4}\right\rceil$ completes the proof.
\end{proof}

\begin{proof}[Proof of Proposition \ref{T:doscf-gaussian-sphere}]
The proof is a straightforward modification of the one in \cite{KLW},
making use of Lemma \ref{T:cosines-average} instead of the corresponding computation
for i.i.d.\ Gaussian coefficients; below are the details for the
necessary modifications, with the part of the proof which is identical
to that of \cite{KLW} omitted.

Suppose that $n$ is even, and make the definitions
\[A:=\sum_{\substack{j=1\\j\,\text{even}}}^n\sum_{a,b=1}^3x_{a,b,j}\sigma^{(a)}_j\sigma^{(b)}_{j+1}\qquad
B:=\sum_{\substack{j=1\\j\,\text{odd}}}^n\sum_{a,b=1}^3x_{a,b,j}\sigma^{(a)}_j\sigma^{(b)}_{j+1}\]
\[A_{3(b-1)+a}:=\sum_{\substack{j=1\\j\,\text{even}}}^nx_{a,b,j}\sigma^{(a)}_j\sigma^{(b)}_{j+1}\qquad
B_{3(b-1)+a}:=\sum_{\substack{j=1\\j\,\text{odd}}}^nx_{a,b,j}\sigma^{(a)}_j\sigma^{(b)}_{j+1}.\]
Then the terms within each sum of each of the $A_k$ and $B_k$ commute,
and 
\[H_n=A+B=\sum_{k=1}^{9}(A_k+B_k).\]
Define 
\[\phi_n(t):=\E\left[\frac{1}{2^n}\tr\left(\prod_{k=1}^9e^{itA_k}e^{itB_k}\right)\right].\]
Observe that since all of the terms within the $A_k$ and $B_k$
commute, 
\[e^{itA_k}=\prod_{\substack{j=1\\j\text{even}}}^ne^{itx_{a,b,j}\sigma^{(a)}_j\sigma^{(b)}_{j+1}}\qquad\qquad
e^{itB_k}=\prod_{\substack{j=1\\j\text{odd}}}^ne^{itx_{a,b,j}\sigma^{(a)}_j\sigma^{(b)}_{j+1}},\]
where $k=3(b-1)+a$.
Now, since the square of any of the Pauli matrices is the identity, it
follows from the definition of the matrix exponential in terms of
power series that
\[e^{itx_{a,b,j}\sigma^{(a)}_j\sigma^{(b)}_{j+1}}=\cos(tx_{a,b,j})I_{2^n}+\sin(tx_{a,b,j})\sigma^{(a)}_j\sigma^{(b)}_{j+1}.\]
By the symmetry of the uniform distribution on $\mathbb{S}^{9n-1}$,
any term with sine factors in the expansion of $\phi_n(t)$ has vanishing
expectation, and so by Lemma \ref{T:cosines-average},
\[\phi_n(t)=\E\left[\prod_{j=1}^n\prod_{a,b=1}^3\cos(tx_k)\right]\xrightarrow{n\to\infty}e^{-\frac{t^2}{2}}.\]

At this point the proof can be completed essentially identically to
the proof in \cite{KLW}.
\end{proof}

From Proposition \ref{T:doscf-gaussian-sphere}, it follows that
$\mu_n^{DOS}$ converges weakly to the standard Gaussian distribution;
since the bounded-Lipschitz distance is a metric for weak convergence,
this means that 
\[\lim_{n\to\infty}d_{BL}(\mu_n^{DOS},\gamma)=0.\]
It follows from part \ref{P:as-dist-to-gauss} of Theorem
\ref{T:random-dist-to-gauss-sphere} that 
\[\lim_{n\to\infty}d_{BL}(\mu_n,\mu_n^{DOS})=0\]
almost surely, and so Theorem \ref{T:sphere-as-convergence} follows.

\section{Concluding remarks}\label{S:remarks}
\begin{enumerate}[label=\arabic*. ,leftmargin=.2in]
\item In \cite{KLW}, the authors consider more general distributional assumptions on the coefficients $\{Z_{a,b,j}\}$.  Specifically, they consider the random matrices
\[H_n:=\sum_{a,b=1}^3\sum_{j=1}^n\alpha_{a,b,j}\sigma^{(a)}_j\sigma^{(b)}_{j+1},\]
where the $\{\alpha_{a,b,j}\}$ are assumed to be independent and symmetric about 0, with
\begin{equation}\label{E:lyapounov-conditions}
\sum_{a,b=1}^3\sum_{j=1}^n\E\alpha_{a,b,j}^2=1\qquad\lim_{n\to\infty}\sum_{a,b=1}^3\sum_{j=1}^n\E|\alpha_{a,b,j}|^{2+\delta}=0\qquad\max_{a,b,j}\E|\alpha_{a,b,j}|^2=o\left(\frac{1}{\sqrt{n}}\right),\end{equation}
for some $\delta>0$.  The point is that these are the conditions on the $\alpha_{a,b,j}$ under which the Lyapounov central limit holds, which allows the authors to show that the pointwise difference between the characteristic function of $\mu_n^{DOS}$ and that of the Gaussian distribution still tends to zero.  
In order to obtain the same rates of convergence as in the Gaussian case, slightly stronger assumptions on the rate of growth of moments are needed; however, without further assumptions, the concentration arguments used to move to almost sure convergence need not apply.

A probability measure $\nu$ is said to satisfy a quadratic transportation cost inequality with constant $a$ if for all probability measures $\mu$ absolutely continuous with respect to $\nu$,
\begin{equation}\label{E:quad-trans-cost}
W_2(\mu,\nu)\le\sqrt{aH(\mu\big\|\nu)},
\end{equation}
where $W_2$ is the $L_2$-Kantorivich distance and $H(\mu\big\|\nu)$ is  the relative entropy (or Kullback-Leibler divergence) of $\mu$ with respect to $\nu$. 
If instead of \eqref{E:lyapounov-conditions}, one assumes that the distribution of each $\alpha_{a,b,j}$  satisfies a quadratic transportation cost inequality with the same constant $a$, then one can carry out the entire program used here in the Gaussian case with essentially no modification, and one again obtains the almost sure convergence of the spectral measure to Gaussian, with a rate of $\frac{1}{n^{1/6}}$ in $W_1$-distance. 
The assumption that the coefficients $\alpha_{a,b,j}$ all share this property is the most general setting of independent $\alpha_{a,b,j}$ in which the arguments using concentration of measure can be carried out (see \cite{Goz} for a detailed discussion).

\item In \cite{ES}, Erd\H{o}s and Schr\"oder introduced a model for quantum spin glasses with arbitrary coupling geometry.
Given a sequence of undirected graphs $\Gamma_n$ on the vertex sets $\{  1, \ldots, n \}$ they considered Hermitian random matrices defined by
$$
H_n^{\Gamma_n} (\bold{x}) = \frac{1}{3 \sqrt{e(\Gamma_n)}} \sum_{(ij)\in \Gamma_n} \sum_{a,b=1}^3 \alpha_{a,b,(ij)}
\sigma_i^{(a)} \sigma_j^{(b)},
$$
where the $\alpha_{a,b,(ij)}$ are are assumed to be independent centered random variables with unit variance, and $e(\Gamma_n)$ denotes the number of edges in $\Gamma_n$.
They proved weak convergence of the density of states measure for this model to the standard Gaussian distribution whenever the maximal degree of a vertex in $\Gamma_n$ is negligible in comparison to $e(\Gamma_n)$

A tail bound similar to Theorem \ref{T:random-dist-to-gauss}, part \ref{P:dist-conc-gauss} can be obtained for this model if the coefficients are assumed to be standard Gaussian random variables.
Using a proof identical to the proof of Lemma 9, one can show that for $\bold{x} =\{ x_{a,b,(ij)} \in \mathbb{R}^{9e(\Gamma_n)} \}$, the map 
$$\bold{x} \rightarrow d_{BL}(\mu_n, \rho)$$
is $\frac{1}{3 \sqrt{e(\Gamma_n)}}$-Lipschitz for any probability measure $\rho$. It then follows from the concentration of measure for standard Gaussian random variables that
$$
\mathbb{P}[|d_{BL}(\mu_n, \rho) - \mathbb{E} d_{BL} (\mu_n, \rho)| > t] \le C e^{- c e(\Gamma_n) t^2}.
$$

Applying this estimate in particular when $\rho$ is the density of states measure, the
almost sure convergence of $\mu_n$ to the standard Gaussian distribution follows from the Borel-Cantelli Lemma and the convergence proved in \cite{ES}.

Erd\H{os} and Schr\"oder also described a model for the Hamiltonian of a quantum $p$-spin glasses: 
$$
H_n^{(p_n-glass)}= 3^{-p_n/2} {n \choose p_n}^{-1/2} 
\sum_{1 \le i_1 < \ldots< i_{p_n} \le n} \sum_{a_1, \ldots, a_{p_n} =1}^3 
\alpha_{a_1,\ldots,a_{p_n},(i_1 \ldots i_{p_n})} \sigma_{i_1}^{(a_1)} \ldots  \sigma_{i_{p_n}}^{(a_{p_n})}.
$$

They found a sharp phase transition at the threshold $p=\sqrt{n}$ between  the  standard Gaussian distribution and the Wigner semicircle law, with an explicitly described limiting measure at criticality.

Using the same arguments as above, it can be shown that for 
$\bold{x} =\{ x_{a,b,(ij)} \in \mathbb{R}^{9e(\Gamma_n)} \}$, the map 
$$\bold{x} \rightarrow d_{BL}(\mu_n, \rho)$$
is $ 3^{-p/2} {n \choose p}^{-1/2} $-Lipschitz for any probability measure $\rho$. 
Thus, if $p_n<<\sqrt{n}$ the empirical spectral measure converges almost surely to the standard Gaussian distribution and if $p_n>>\sqrt{n}$ the empirical spectral measure converges almost surely to the semicircle law, while if $\frac{p_n}{\sqrt{n}}\xrightarrow{n\to\infty}\lambda\in(0,\infty)$, the empirical spectral measure converges almost surely to the limiting measure described in \cite{ES}.

\end{enumerate}

\section*{Acknowledgements}
This research was supported by  grant DMS-1308725 from the National Science Foundation.  Some of the work was
carried out while the second-named author was visiting the Institut de
Math\'ematiques de Toulouse at the Universit\'e Paul Sabatier.


\bibliographystyle{plain}
\bibliography{qspins}

\end{document}